\documentclass[11pt]{article}
\usepackage{amsmath,amssymb,amsthm,amsfonts}
\usepackage{mathtools}
\usepackage{fullpage,color}
\usepackage{times,microtype,hyperref,pdfsync}
\usepackage{algorithm}
\usepackage[noend]{algorithmic}
\usepackage{enumitem}
\usepackage{xspace}

\newcommand{\abs}[1]{\lvert{#1}\rvert}

\newcommand{\set}[1]{\{{#1}\}}
\newcommand{\setfit}[1]{\left\{{#1}\right\}}

\newcommand{\parensfit}[1]{\left({#1}\right)}

\newcommand{\length}[1]{\lVert{#1}\rVert}

\theoremstyle{plain}            
\newtheorem{theorem}{Theorem}[section]
\newtheorem{lemma}[theorem]{Lemma}
\newtheorem{corollary}[theorem]{Corollary}

\newtheorem{fact}[theorem]{Fact}

\theoremstyle{definition}       
\newtheorem{definition}[theorem]{Definition}

\theoremstyle{remark}           
\newtheorem{remark}[theorem]{Remark}

\numberwithin{equation}{section}

\newif\ifnotes\notesfalse

\ifnotes

\newcommand{\cnote}[1]{{\bf (Chris:} {#1}{\bf ) }}
\newcommand{\enote}[1]{{\bf (Elena:} {#1}{\bf ) }}

\else

\newcommand{\cnote}[1]{}
\newcommand{\enote}[1]{}

\fi

\newcommand{\Z}{\mathbb{Z}}
\newcommand{\G}{\mathbb{G}}
\newcommand{\C}{\mathbb{C}}
\newcommand{\F}{\mathbb{F}}
\newcommand{\R}{\mathbb{R}}
\newcommand{\RM}[2]{\text{RM}^{#1}_{#2}}
\newcommand{\BW}{\text{BW}}
\newcommand{\eps}{\epsilon}
\newcommand{\call}{\mathcal{L}}
\newcommand{\calt}{\mathcal{T}}
\newcommand{\calc}{\mathcal{C}}

\newcommand{\rsd}{\ensuremath{\mathsf{rsd}}\xspace}
\newcommand{\rsmd}{\ensuremath{\mathsf{rsmd}}\xspace}
\newcommand{\bit}{\ensuremath{\set{0,1}}}

\DeclareMathOperator{\poly}{poly}
\DeclareMathOperator{\polylog}{polylog}

\title{List Decoding Barnes-Wall Lattices}

\author{Elena Grigorescu\thanks{School of Computer Science, Georgia
    Institute of Technology.  Email: \texttt{elena\_g@csail.mit.edu}.
    This material is based upon work supported by the National Science
    Foundation under Grant \#1019343 to the Computing Research
    Association for the CI Fellows Project.}  \and Chris
  Peikert\thanks{School of Computer Science, Georgia Institute of
    Technology.  Email: \texttt{cpeikert@cc.gatech.edu}.  This
    material is based upon work supported by the National Science
    Foundation under CAREER Award~CCF-1054495 and the Alfred P.~Sloan
    Foundation.  The views expressed are those of the authors and do
    not necessarily reflect the official policy or position of the
    National Science Foundation of the Sloan Foundation.}}

\begin{document}

\maketitle
\def\conf{0} 

\thispagestyle{empty}
\begin{abstract}
  \iffalse

\else
The question of \emph{list decoding} error-correcting codes over
finite fields (under the Hamming metric) has been widely studied in
recent years.  Motivated by the similar discrete linear structure of
linear codes and \emph{point lattices} in $\R^{N}$, and their many
shared applications across complexity theory, cryptography, and coding
theory, we initiate the study of list decoding for lattices.  Namely:
for a lattice $\call\subseteq \R^N$, given a target vector $r\in \R^N$
and a distance parameter $d$, output the set of all lattice points $w
\in \call$ that are within distance $d$ of $r$.

In this work we focus on combinatorial and algorithmic questions
related to list decoding for the well-studied family of
\emph{Barnes-Wall} lattices.  Our main contributions are twofold:
\begin{enumerate}
\item We give tight (up to polynomials) combinatorial bounds on the
  worst-case list size, showing it to be polynomial in the lattice
  dimension for any error radius bounded away from the lattice's
  minimum distance (in the Euclidean norm).
\item Building on the \emph{unique} decoding algorithm of Micciancio
  and Nicolosi (ISIT '08), we give a list-decoding algorithm that runs
  in time polynomial in the lattice dimension and worst-case list
  size, for any error radius.  Moreover, our algorithm is highly
  parallelizable, and with sufficiently many processors can run in
  parallel time only \emph{poly-logarithmic} in the lattice dimension.
\end{enumerate}
In particular, our results imply a polynomial-time list-decoding
algorithm for any error radius bounded away from the minimum distance,
thus beating a typical barrier for natural error-correcting codes
posed by the Johnson radius.  \fi


\end{abstract}

\newpage

\setcounter{page}{1}

\section{Introduction}

A linear error-correcting \emph{code} $\calc$ of block length $N$ and
dimension $K$ over a field $\F$ is a $K$-dimensional subspace of
$\F^{N}$, generated as all $\F$-linear combinations of $K$ linearly
independent vectors.  The code's \emph{minimum distance}, denoted
$d(\calc)$, is the minimum Hamming distance between any two distinct
codewords in $\calc$, or equivalently the minimum Hamming weight over
all nonzero codewords.  It is often convenient to normalize distances
by the dimension, yielding the \emph{relative} (minimum) distance
$\delta(\calc)=d(\calc)/N$ of the code.  Similarly, a point
\emph{lattice} of dimension~$N$ and rank~$K$ (where often $K=N$) is a
discrete additive subgroup of~$\R^{N}$ (or~$\C^{N}$), generated as all
integer linear combinations of $K$ linearly independent vectors.  The
lattice's minimum distance $\lambda(\call)$ is the minimum Euclidean
norm over all nonzero lattice points $x \in \call$.  Here it can also
be convenient to normalize by the dimension, and for a closer analogy
between the Hamming and Euclidean distances, in what follows we work
with the \emph{relative squared} distance (abbreviated \rsd)
$\delta(x,y)=\delta(x-y)$ on $\R^{N}$ or $\C^{N}$, where
$\delta(z)=\frac{1}{N} \length{z}^{2}=\frac{1}{N} \sum_{i=1}^{N}
\abs{z_{i}}^{2}$.  The {\em relative squared minimum distance} (abbreviated
$\rsmd$) $\delta(\call)$ of a lattice is therefore
$\delta(\call)=\lambda(\call)^{2}/N$.

Codes and lattices are intensely studied objects, with many
applications in computational complexity, cryptography, and coding
theory.  In particular, both kinds of objects can be used to encode
data so that it can be recovered reliably after being sent over a
noisy channel.  A central question associated with codes is
\emph{unique decoding}: given a received word $r \in \F^{N}$ within
relative Hamming distance less than ${\delta(\calc)}/{2}$ of some
codeword $w \in \calc$, find $w$.  Similarly, the unique (also known
as bounded-distance) decoding problem on lattices is: given a received
word $r \in \R^{N}$ within \rsd less than ${\delta(\call)}/{4}$ of
some lattice vector $v \in \call$, find $v$.  (Note that the $1/4$
factor arises because distances are squared in our formulation.)

For error-correcting codes, Elias~\cite{elias} and
Wozencraft~\cite{wozencraft} proposed extending the classical unique
decoding problem to settings where the amount of error could cause
ambiguous decoding.  More precisely, the goal of \emph{list decoding}
is to find all codewords within a certain relative distance (typically
exceeding ${d(\calc)}/{2}$) of a received word; in many cases, the
list is guaranteed to contain few codewords.  The first breakthrough
algorithmic list decoding results were due to Goldreich and
Levin~\cite{DBLP:conf/stoc/GoldreichL89} for the Hadamard code, and to
Sudan~\cite{Sudan97} and Guruswami-Sudan~\cite{GuSu} for Reed-Solomon
codes.  These results and others have had countless applications,
e.g., in building hard-core predicates for one-way
functions~\cite{DBLP:conf/stoc/GoldreichL89}, in hardness
amplification~\cite{SudanTV01-journal}, in learning Fourier
coefficients~\cite{KM-fourier,GGIMS,AGS}, and in constructing
randomness
extractors~\cite{tashma-zuckerman,Trevisan01,GuruswamiUV09}.

There are two central tasks associated with list decoding:
combinatorially bounding the number of codewords within a given radius
of a received word, and algorithmically finding these codewords.  An
important question in understanding list decodability is finding the
\emph{list-decoding radius} of the code, i.e., the maximum distance
from a received word within which the number of codewords is
guaranteed to be polynomial in the input parameters.

\ifnum\conf=2
\smallskip
{\bf The Johnson bound.}
\else
\vspace{-8pt}
\paragraph{The Johnson bound.}
\fi
Under the Hamming metric, the \emph{Johnson bound} gives a distance up
to which list decoding is guaranteed to be combinatorially efficient.
One version of the Johnson bound states that for any code $\calc$ of
relative distance $\delta$, a Hamming ball of
relative radius $J(\delta)-\eps$ contains at most $1/\epsilon^{2}$
codewords, and a ball of relative radius $J(\delta)$ contains at most
$\delta N^{2} \abs{\F}$ codewords, where
$J(\delta)=1-\sqrt{1-\delta}$.  The Johnson bound is generic since it
does not use any structure of the code (not even linearity), and in
many cases it is not necessarily the same as the list-decoding radius.
It is, however, a barrier in the current analysis of combinatorial
list decoding for many well-studied families like Reed-Solomon codes,
algebraic geometry codes, Chinese remainder codes, and others.  The
breakthrough works of Parvaresh-Vardy~\cite{PV} and
Guruswami-Rudra~\cite{DBLP:conf/stoc/GuruswamiR06} gave families of
codes which could be (efficiently) list decoded beyond the Johnson
bound, and were followed by several related combinatorial and
algorithmic results for other codes
(e.g.,~\cite{DGKS08,GopalanKZ08,KaufmanLP10,GopalanGR09}).  For more
detailed surveys on list decoding of codes we refer to
\cite{Sudan00,Guruswami2004,Guruswami06, Guruswamibridgingshannon}.

\subsection{Contributions}
\label{sec:contributions}

Motivated by the common discrete linear structure of codes and
lattices, we initiate the study of list decoding for lattices, from
both a combinatorial and algorithmic perspective.
Conway and Sloane~\cite{ConwayS98} promoted the applicability of
lattices in practice as alternatives to codes.  Therefore, our study
is motivated by practical applications in error-tolerant
communication, but primarily by the naturalness of the list-decoding
problem from a mathematical and computational perspective, and we hope
that our work will find other applications in theoretical computer
science.

In this work we focus on the \emph{Barnes-Wall} (BW)~\cite{BW59}
family of lattices in $\C^{N}$, which have been well-studied in coding
theory (see,
e.g.,~\cite{Forney88a,ForneyV96,amir98:_trell_compl,NebeRS01,SalomonA05})
and share many connections to the Reed-Muller~\cite{RMuller,ReedM}
family of error-correcting codes (we elaborate below).  Barnes-Wall
lattices were first constructed in order to demonstrate dense sphere
packings, a feature that makes them useful in communications settings.
Minimum-distance decoding algorithms for BW lattices were given
in~\cite{Forney88a,ran98:_effic_decod_of_gosset_coxet,WangST95}, but
they are either for fixed low dimensions or have runtimes exponential
in the lattice dimension~$N$.  Micciancio and Nicolosi~\cite{MN08}
gave the first $\poly(N)$-time algorithms for bounded-distance
(unique) decoding of any BW lattice up to $\delta/4$ relative error,
along with parallel versions which run in as little as $\polylog(N)$
parallel time on sufficiently many processors.  They also posed list
decoding of BW lattices as an open problem.

Our main contributions are twofold:
\begin{enumerate}
\item We give tight (up to polynomials) combinatorial bounds on the
  worst-case list size for BW lattices, showing it to be polynomial in
  the lattice dimension $N$ for any relative squared distance (\rsd)
  bounded away from the \rsmd $\delta$ of the lattice.  (See
  Theorems~\ref{thm:listbound} and~\ref{thm:lower-bound} below for
  precise statements.)  We note that it was already known that the
  list size is super-polynomial $N^{\Theta(\log N)}$ when the \rsd
  equals $\delta$ (see, e.g.,~\cite[Chapter 1, \S 2.2, page
  24]{ConwayS98}).

\item We give a corresponding list-decoding algorithm that, for any
  \rsd, runs in time polynomial in the lattice dimension and
  worst-case list size.  Our algorithm is a variant of the
  Micciancio-Nicolosi unique-decoding algorithm, and as such it is
  also highly parallelizable: with sufficiently many processors it
  runs in only poly-logarithmic $O(\log^{2} N)$ parallel time.
\end{enumerate}

We note that Johnson-type bounds for lattices are known and easy to
obtain (in fact, the Johnson bound for codes under the Hamming metric
is typically proved by reducing it to a packing bound in $\R^{N}$
under the Euclidean norm; see, e.g.,~\cite{Bollobas, GS-johnson,
  madhulectnotes, MGbook02}).  For a lattice $\call \subset \C^{N}$
with \rsmd $\delta$, the list size for \rsd $\delta \cdot
(\tfrac{1}{2}-\eps)$ is at most $\frac{1}{2\eps}$, and for \rsd
$\frac{\delta}{2}$ is at most $4N$ (see Lemma~\ref{lem:johnson}).
Interestingly, the latter bound is tight for BW lattices (see
Corollary~\ref{cor:johnson}).  Since $\delta = 1$ for every BW
lattice, our combinatorial and algorithmic results for \rsd up to $1$
therefore apply far beyond the Johnson bound.

To describe our results in more detail, we need to define Barnes-Wall
lattices.  Let $\G=\Z[i]$ be the ring of Gaussian integers, and let
$\phi = 1+i\in \G$.
\begin{definition}[Barnes-Wall lattice]
  \label{def:BW}
  The $n$th Barnes-Wall lattice $\BW_{n} \subseteq \G^{N}$ of
  dimension $N=2^n$ is defined recursively as $\BW_{0} = \G$, and for
  positive integer $n \geq 1$ as
  \[ \BW_{n}=\setfit{ [u, u+\phi v] : u,v\in \BW_{n-1} }. \]
\end{definition}
One can check that $\BW_{n}$ is a lattice; indeed, it is easy to
verify that it is generated as the $\G$-linear combinations of the
rows of the $n$-fold Kronecker product 
\[ W=\begin{bmatrix} 1 & 1 \\ 0 & \phi \end{bmatrix}^{\otimes n} \in
\C^{N \times N}. \] A simple induction proves that the minimum
distance of $\BW_{n}$ is $\sqrt{N}$, i.e., its \rsmd is
$\delta=1$.\footnote{The fundamental volume of $\BW_{n}$ in $\C^{N}$
  is $\det(W)=2^{nN/2}$, so its determinant-normalized minimum
  distance is $\sqrt{N}/\det(W)^{1/(2N)}=\sqrt[4]{N}$.  This is better
  than the normalized minimum distance $1$ of the integer lattice
  $\G^{N}$, but worse that the largest possible of $\Theta(\sqrt{N})$
  for $N$-dimensional lattices.}  Also observe that if $[u, w=u+\phi
v]\in \BW_{n}$ for $u,w \in \C^{N/2}$, then $[w, u]\in \BW_{n}$:
indeed, we have $w,-v \in \BW_{n-1}$ and so $[w,u=w + \phi \cdot -v]
\in \BW_{n}$.  The mathematical and coding properties of Barnes-Wall
lattices have been studied in numerous works,
e.g.,~\cite{AgrawalV00,ConwayS98,Forney88a,ForneyV96,NebeRS01,SalomonA05,WangST95,MN08}.


\ifnum\conf=2
\smallskip
{\bf Combinatorial bounds.}
\else
\paragraph{Combinatorial bounds.}
\fi Let $\ell(\eta, n)$ denote the worst-case list size (over all
received words) for $\BW_{n}$ at \rsd $\eta$.  We prove the following
upper bound.

\begin{theorem}
  \label{thm:listbound}
  For any integer $n \geq 0$ and real $\epsilon > 0$, we have
  \[ \ell(1-\epsilon, n) \leq
  4\cdot(1/\epsilon)^{16n}=N^{O(\log(1/\epsilon))}. \]
\end{theorem}

\noindent Moreover, we show that the above bound is tight, up to
polynomials.

\begin{theorem}
  \label{thm:lower-bound}
  For any integer $n \geq 0$ and $\eps \in [2^{-n}, 1]$, we have
  \[ \ell(1-\epsilon, n) \geq 2^{(n-\log
    \frac{1}{\eps})\log\frac{1}{2\eps}}. \] In particular, for any
  constant $\epsilon > 0$ (or even any $\eps \geq N^{-c}$ for $c <
  1$), we have $\ell(1-\epsilon, n) = N^{\Omega(\log(1/\epsilon))}$.
\end{theorem}

As previously mentioned, it is also known that at \rsd $\eta=1$, the
maximum list size $\ell(1,n)$ is quasi-polynomial $N^{\Theta(\log N)}$
in the lattice dimension, and is achieved by letting the received word
be any lattice point~\cite[Chapter 1, \S 2.2, page 24]{ConwayS98}.
Because the \rsmd of $\BW_{n}$ is exactly $1$, here we are just
considering the number of lattice points at minimum distance from the
origin, the so-called ``kissing number'' of the lattice.

\ifnum\conf=2
\smallskip
{\bf List-decoding algorithm.}
\else
\paragraph{List-decoding algorithm.}
\fi
We complement the above combinatorial bounds with an algorithmic
counterpart, which builds upon the unique (bounded-distance) decoding
algorithm of Micciancio and Nicolosi~\cite{MN08} for \rsd up to
$\frac14$.

\begin{theorem}
  \label{thm:bwlist-alg}
  There is a deterministic algorithm that, given any received word
  $r\in \C^{N}$ and $\eta\geq 0$, outputs the list of all points in
  $\BW_{n}$ that lie within \rsd $\eta$ of $r$, and runs in time
  $O(N^2) \cdot \ell(\eta, n)^{2}$.
\end{theorem}

We also remark that the algorithm can be parallelized just as
in~\cite{MN08}, and runs in only polylogarithmic $O(\log^2 N)$
parallel time on $p \geq N^{2} \cdot \ell(\eta, n)^2$ processors.

\medskip

\noindent Theorems~\ref{thm:listbound} and \ref{thm:bwlist-alg}
immediately imply the following corollary for $\eta=1-\epsilon$.

\begin{corollary}
  \label{cor:bwlist-alg}
  There is a deterministic algorithm that, given a received word $r\in
  \C^{N}$ and $\eps > 0$, outputs the list of all lattice points in
  $\BW_{n}$ that lie within \rsd $(1-\eps)$ of $r$, and runs in time
  $(1/\epsilon)^{O(n)}=N^{O(\log(1/\epsilon))}$.
\end{corollary}

Given the lower bounds, our algorithm is optimal in the sense that for
any constant $\epsilon > 0$, it runs in $\poly(N)$ time for \rsd
$1-\epsilon$, and that list decoding in $\poly(N)$ time is impossible
(in the worst case) at \rsd~$1$.

\subsection{Proof Overview and Techniques}
\label{sec:overview}


\ifnum\conf=2
{\bf Combinatorial bounds.}
\else
\paragraph{Combinatorial bounds.}
\fi Our combinatorial results exploit a few simple observations, some
of which were also useful in obtaining the algorithmic results
of~\cite{MN08}.  The first is that by the Pythagorean theorem, if
$\eta = \delta(r,w)$ is the \rsd between a received vector $r=[r_0,
r_1]\in \C^N$ and a lattice vector $w=[w_0, w_1] \in \BW_{n}$ (where
$r_{i} \in \C^{N/2}$ and $w_{i} \in \BW_{n-1}$), then
$\delta(r_{b},w_{b}) \leq \eta$ for some $b \in \bit$.  The second
observation (proved above) is that BW lattices are closed under the
operation of swapping the two halves of their vectors, namely, $[w_0,
w_{1}] \in \BW_{n}$ if and only if $[w_{1}, w_{0}] \in \BW_{n}$.
Therefore, without loss of generality we can assume that
$\delta(r_{0}, w_{0}) \leq \eta$, while incurring only an extra factor
of $2$ in the final list size.  A final important fact is the
relationship between the \rsd's for the two Barnes-Wall vectors
$u=w_{0}, v=\frac{1}{\phi}(w_{1}-w_{0}) \in \BW_{n-1}$ that determine
$w$; namely, we have \[ \eta = \tfrac12 \delta(r_{0}, u) +
\delta(\tfrac{1}{\phi}(r_{1}-u), v). \] (See
Lemma~\ref{lem:delta-relation}.)  Since $\delta(r_{0}, u) \leq \eta$,
we have must have $\delta(\frac{1}{\phi}(r_{1}-u), v) = \eta - \frac12
\delta(r_{0}, w_{0}) \in [\eta/2,\eta]$.

Our critical insight in analyzing the list size is to carefully
partition the lattice vectors in the list according to their distances
from the respective halves of the received word.  Informally, a larger
distance on the left half (between $r_{0}$ and $u$) allows for a
larger list of $u$'s, but also implies a smaller distance on the right
half (between $\frac{1}{\phi}(r_{1}-u)$ and $v$), which limits the
number of possible corresponding $v$'s.  We bound the total list size
using an inductive argument for various carefully chosen ranges of the
distances at lower dimensions.  Remarkably, this technique along with
the Johnson bound allows us to obtain tight combinatorial bounds on
the list size for distances all the way up to the minimum distance.

As a warm-up example, which also serves as an important step when
analyzing larger \rsd's, Lemma~\ref{lem:5/8} gives a bound of
$\ell(\tfrac58, n) \leq 4 \cdot 24^{n} = \poly(N)$ for \rsd
$\eta=\tfrac58$.  This bound is obtained by partitioning according to
the two cases $\delta(r_{0}, u) \in [0,\frac{5}{12})$ and
$\delta(r_{0}, u) \in [\frac{5}{12}, \frac{5}{8}]$, which imply that
the \rsd between $v$ and $\frac{1}{\phi}(r_{1}-u)$ is at most
$\frac{5}{8}$ and $\frac{5}{12}$, respectively.  When bounding the
corresponding number of $u$'s and $v$'s, the \rsd's up to
$\frac{5}{12} < \frac12$ are handled by the Johnson bound, and \rsd's
up to $\frac{5}{8}$ are handled by induction on the dimension.

To extend the argument to \rsd's up to $\eta = 1-\epsilon$, we need to
partition into three cases, including ones which involve \rsd's
$1-\frac{3\epsilon}{2}$ and $\frac34$.  In turn, the bound for \rsd
$\frac34$ also uses three cases, plus the above bound for \rsd
$\frac58$.  Interestingly, all our attempts to use fewer cases or a
more direct analysis resulted in qualitatively worse list size bounds,
such as $N^{O(\log^{2}(1/\epsilon))}$ or worse.

Lastly, our lower bounds from Theorem~\ref{thm:lower-bound} are
obtained by using a representation of BW lattices in terms of RM codes
(see Fact~\ref{fact:BW-RM}), and by adapting the lower bounds
from~\cite{GopalanKZ08} for RM codes to BW lattices.

\ifnum\conf=2
\smallskip
{\bf List-decoding algorithm.}
\else
\paragraph{List-decoding algorithm.}
\fi A natural approach to devising a list-decoding algorithm using the
above facts (also used in the context of Reed-Muller
codes~\cite{GopalanKZ08}) is to first list decode the left half
$r_{0}$ of the received word to get a list of $u$'s, and then
sequentially run through the output list to decode the right half
$\frac1{\phi}(r_{1}-u)$ and get a corresponding list of $v$'s for each
value of $u$.  However, because the recursion has depth $n$, the
straightforward analysis reveals a super-polynomial runtime
$N^{\Omega(n)}$ for \rsd $\eta \geq 1/2$, because the list size at
depth $d$ can be $\geq 4N/2^{d}$.

Instead, our list-decoding algorithm is based on the elegant
divide-and-conquer algorithm of~\cite{MN08} for bounded-distance
(unique) decoding, which decodes up to half the minimum distance
(i.e., $\eta=\frac14$) in quasi-linear $\tilde{O}(N)$ time, or even
poly-logarithmic $O(\log^{c} N)$ parallel time on a sufficiently large
$\poly(N)$ number of processors.

The main feature of the algorithm, which we exploit in our algorithm
as well, is the use of a distance-preserving linear automorphism
$\calt$ of the BW lattice, i.e., $\calt(\BW_{n}) = \BW_{n}$ (see
Fact~\ref{fact:automorphisms}).  In particular, a lattice vector $w
\in \BW_{n}$ can be reconstructed from just \emph{one} arbitrary half
of each of $w = [w_{0}, w_{1}]$ and $\calt(w) = [\calt_{0}(w),
\calt_{1}(w)]$.  Recall that for a received word $r=[r_{0},r_{1}]$
(where $r_{i} \in \C^{N/2}$), we are guaranteed that $\delta(r_{b},
w_{b}) \leq \delta(r,w)$ for some $b \in \bit$, and similarly for
$\calt(r)$ and $\calt(w)$.  These facts straightforwardly yield a
divide-and-conquer, parallelizable list-decoding algorithm that
recursively list decodes each of the four halves $r_{0}, r_{1},
T_{0}(r), T_{1}(r)$ and reconstructs a list of solutions by combining
appropriate pairs from the sub-lists, and keeping only those that are
within the distance bound.  The runtime of this algorithm is only
quadratic in the worst-case list size, times a $\poly(N)$ factor (see
Section~\ref{sec:list-decod-alg}).  We emphasize that the only
difference between our algorithm and the MN algorithm is the simple
but crucial observation that one can replace single words by lists in
the recursive steps.  The runtime analysis, however, is entirely
different, because it depends on the combinatorial bounds on list
size.

\subsection{Comparison with Reed-Muller Codes} 
\label{sec:comparison-RM}

Here we discuss several common and distinguishing features of
Barnes-Wall lattices and Reed-Muller codes.

\begin{definition}[Reed-Muller code]
  \label{def:RM}
  For integers $d,n \geq 0$, the Reed-Muller code of degree $d$ in $n$
  variables (over $\F_{2}$) is defined as
  \[ \RM{d}{n}=\setfit{ \langle p(\alpha) \rangle_{\alpha \in \F_2^n}
    : p \in\F_2[x_1, \ldots, x_n], \deg(p) \leq d}. \] An
  equivalent recursive definition is $\RM{0}{n} = \set{\bar{0},
    \bar{1}} \subseteq \F_{2}^{2^{n}}$ for any integer $n \geq 0$,
  and \[ \RM{d}{n}=\setfit{ [u, u+v] : u\in \RM{d}{n-1}, v\in
  \RM{d-1}{n-1}}. \] Here if $u \in \RM{d}{n-1}, v \in \RM{d-1}{n-1}$
  correspond to polynomials $p_{u}, p_{v} \in \F_{2}[x_{1}, \ldots
  x_{n-1}]$ respectively, then the codeword $[u, u+v] \in \RM{d}{n}$
  corresponds to the polynomial $p = p_{u} + x_{n} \cdot p_{v} \in
  \F_{2}[x_{1}, \ldots, x_{n}]$.
\end{definition}

The recursive definition of RM codes already hints at structural
similarities between BW lattices and RM codes.  Indeed, BW lattices
can be equivalently defined as evaluations modulo $\phi^{n}$ of
(Gaussian) integer multilinear polynomials in $n$ variables over the
domain $\set{0, \phi}^n$.  Recall that an integer multilinear
polynomial $p \in \G[x_{1}, \ldots, x_{n}]$ is one whose monomials
have degree at most one in each variable (and hence total degree at
most $n$), i.e., \[ p(x_1, \ldots, x_n)=\sum_{S\in \bit^n} a_S \cdot
\prod_{i\in S} x_i \] where each $a_{S} \in \G$.  A simple inductive
argument proves the following lemma.
\begin{lemma}
  \label{lem:bw-multilinear}
  $\BW_{n}=\phi^{n} \G^{2^n} + \set{ \langle p(x)\rangle_{x\in \{0,
      \phi\}^n} : p \in \G[x_{1}, \ldots, x_{n}] \text{ is
      multilinear} }$.
\end{lemma}
Thus, while $\RM{d}{n}$ codewords correspond to low-degree polynomials
(when $d$ is small), BW lattice points correspond to possibly
high-degree polynomials.
%
As an immediate application, our main theorems imply the following
corollary regarding the set of integer multilinear polynomials that
approximate a function $f \colon \set{0, \phi}^n \to \C$.

\begin{corollary}
  \label{cor:bw-recover-multilin}
  Given a map $f:\set{0, \phi}^n \to \C$ (represented as a lookup
  table) and $\epsilon=\Omega{(N^{-c})}$ for some $c < 1$ and
  $N=2^{n}$, there exists an algorithm that outputs in time
  $N^{O(\log(1/\epsilon))}$ all the integer multilinear polynomials $g
  \colon \set{0, \phi}^n \to \C$ such that $\length{f-g}^2 \leq
  (1-\eps)N$.
\end{corollary}

Just as in our algorithmic results for BW lattices, the recursive
structure of RM codes is critically used in list-decoding algorithms
for these codes, but in a different way than in our algorithm.  The
list-decoding algorithm for $\RM{d}{n}$ given in~\cite{GopalanKZ08}
recursively list decodes one of the halves of a received word, and
then for each codeword in the list it recursively list decodes the
other half of the received word.  The recursion has depth $d$ and thus
has a total running time of $\poly(N) \cdot \ell(\eta)^{d}$, where
$\ell(\eta)$ is the list size at relative (Hamming) distance~$\eta$.
As mentioned above, a similar algorithm can work for BW lattices, but
the natural analysis implies a super-polynomial $\ell(\eta)^{n}$ lower
bound on the running time, since now the recursion has depth $n$.  The
reason we can overcome this potential bottleneck is the existence of
the linear automorphism $\calt$ of $\BW_{n}$, which allows us to make
only a \emph{constant} number of recursive calls (independently of
each other), plus a $\poly(N) \cdot \ell(\eta)^{2}$-time combining
step, which yields a runtime of the form $O(1)^{n} \cdot \poly(N)
\cdot \ell(\eta)^{2} = \poly(N) \cdot \ell(\eta)^{2}$.

 
We note that $\RM{d}{n}$ codes are efficiently list decodable up to a
radius larger than the minimum distance~\cite{GopalanKZ08}, and remark
that while RM codes are some of the oldest and most intensively
studied codes, it was not until recently that their list-decoding
properties have been very well
understood~\cite{PellikaanW04,GopalanKZ08,KaufmanLP10}.

\subsection{Other Related Work}
\label{sec:other-related-work}

Cohn and Heninger~\cite{cohnHeninger} study a list-decoding model on
polynomial lattices, under both the Hamming metric and certain
`non-Archimedian' norms.  Their polynomial analogue of Coppersmith's
theorem~\cite{Coppersmith01} implies, as a special case, Guruswami and
Sudan's result on list decoding Reed-Solomon codes~\cite{GuSu}.

Decoding and list decoding in the Euclidean space has been also
considered for embeddings into real vector spaces of codes classically
defined over finite fields.  These embeddings can give rise to
so-called spherical codes, where the decoding problem has as input a
received vector on the unit sphere, and is required to output the
points in the code (also on the unit sphere) that form a small angle
with the given target.  Another related decoding model is
soft-decision decoding, where for each position of the received word,
each alphabet symbol is assigned a real-valued weight representing the
confidence that the received symbol matches it.  Soft decision unique
decoding for RM codes was studied
in~\cite{DumerK00,DumerS06,DumerS06-2}, and list-decoding algorithms
were shown in~\cite{DumerKT08,FourquetT08}.

Further, the question of decoding lattices is related to the
well-studied {\em vector quantization} problem.  In this problem,
vectors in the ambient space need to be rounded to nearby points of a
discrete lattice; for further details on this problem see, for
example,~\cite{ConwayS98}.

\ifnum\conf=2
\smallskip
{\bf Organization.}
\else
\paragraph{Organization.}
\fi
\ifnum\conf=1
In Section~\ref{sec:combinatorial-bounds} we prove our combinatorial
upper bounds for BW lattices (due to space restrictions, we defer the
lower bounds to Appendix~\ref{sec:lower-bounds}).  In
Section~\ref{sec:list-decod-alg} we present and analyze our main
list-decoding algorithm.  We conclude with several open problems in
Section~\ref{sec:disc-open-probl}.  All missing proofs may be found in
the appendix.

\else
In Section~\ref{sec:combinatorial-bounds} we prove our combinatorial
upper and lower bounds for BW lattices.  In
Section~\ref{sec:list-decod-alg} we present and analyze our main
list-decoding algorithm.  We conclude with several open problems in
Section~\ref{sec:disc-open-probl}. 

\fi




\section{Combinatorial Bounds}
\label{sec:combinatorial-bounds}

We start with a few basic definitions.  For a lattice $\call$, a
vector $r \in \C^{m}$ (often called a received word) and any $\eta
\geq 0$, define $L_{\call}(r, \eta) = \set{ x \in \call : \delta(r,x)
  \leq \eta} $ to be the list of lattice points $w \in \call$ such
that $\delta(r,w) \leq \eta$.  We often omit the subscript $\call$
when the lattice is clear from context.  For $\eta \geq 0$ and
nonnegative integer $n$ with $N=2^{n}$, we define $\ell(\eta, n) =
\max_{r \in \C^{n}} \abs{L_{\BW_{n}}(r, \eta)}$ to be the maximum list
size for \rsd $\eta$, for the $n$th Barnes-Wall lattice.

\subsection{Helpful Lemmas}

We start with two simple but important observations about Barnes-Wall
lattices.  The first relates the \rsd's between the respective ``left''
and ``right'' halves of a received word and a lattice point.  The
second relates the list sizes for the same \rsd but different
dimensions.

\begin{lemma}
  \label{lem:delta-relation}
  Let $r=[r_0, r_1]\in \C^N$ with $r_0, r_1 \in \C^{N/2}$, and $w=[u,
  u+\phi v]\in \BW_{n}$ for $u,v \in \BW_{n-1}$.  Let
  $\eta=\delta(r, w)$, $\eta_0=\delta(r_0, u)$ and
  $\eta_1=\delta(\frac{1}{\phi} (r_1-u), v)$.  Then
  $\eta=\frac{\eta_0}{2}+\eta_1$.
\end{lemma}

\ifnum\conf=2
\begin{IEEEproof}
  We have
  \begin{align*}
    \delta(r, w) &= \frac{\delta(r_{0},u) + \delta(r_{1}, u+\phi
      v)}{2} \\
    &= \frac{\eta_{0}}{2} + \frac{\abs{\phi}^{2} \cdot
      \delta(\frac{1}{\phi} (r_{1}-u), v)}{2} \\
    &= \frac{\eta_{0}}{2} +
    \eta_{1}. \quad \IEEEQED
  \end{align*} \let\IEEEQED\relax
\end{IEEEproof}
\else
\begin{proof}
  We have
  \[ \delta(r, w) = \frac{\delta(r_{0},u) + \delta(r_{1}, u+\phi
    v)}{2} 
    = \frac{\eta_{0}}{2} + \frac{\abs{\phi}^{2} \cdot
    \delta(\frac{1}{\phi} (r_{1}-u), v)}{2} 
    = \frac{\eta_{0}}{2} +
  \eta_{1}. \qedhere \]  
\end{proof}
\fi

\begin{lemma}
  \label{lem:same-rsd}
  For any $\eta \geq 0$ and $n \geq 1$, we have $\ell(\eta, n-1)
  \leq \ell(\eta, n).$
\end{lemma}

\begin{proof}
  Let $r\in \C^{N/2}$ and $w\in L(r, \eta) \subseteq \BW_{n-1}$.
  Then $\delta([r, r], [w, w])= \delta(r, w)$, and since $[w, w]\in
  \BW_{n}$ (because $w\in \BW_{n-1}$) it follows that $[w,w]\in L([r,
  r], \eta)$.
\end{proof}


We next state a Johnson-type bound on the list size for arbitrary
lattices; see, e.g.,~\cite{Bollobas, GS-johnson, madhulectnotes,
  MGbook02} for proofs.  Note that these sources work in $\R^N$; our
form follows because the standard isomorphism between $\C^N$ and
$\R^{2N}$ as real vectors spaces also preserves Euclidean norm.

\begin{lemma}[{\bf Johnson bound}]
  \label{lem:johnson}
  Let $\call \subset \C^N$ be a lattice of \rsmd
  $\delta=\delta(\call)$ and let $r\in \C^N$.  Then
  \begin{enumerate}
  \item $\abs{L(r, \frac{\delta}{2})} \leq 4N$, and \label{item:johnson-1}
  \item $\abs{L(r, \delta \cdot (\tfrac12-\epsilon))} \leq
    \frac{1}{2\epsilon}$ for any $\epsilon >
    0$. \label{item:johnson-2}
  \end{enumerate}
\end{lemma}

(In reading these bounds, recall that $\delta(\call)/4$, not
$\delta(\call)/2$, is the relative unique-decoding distance of
$\call$, because $\delta(\call)$ is the relative \emph{squared}
minimum distance of the lattice.)


\begin{corollary}
  \label{cor:johnson}
  For the lattice $\BW_{n} \subseteq \C^{N}$ and any $\epsilon > 0$,
  we have $\ell(\frac12, n)=4N$ and $\ell(\frac12-\eps, n) \leq
  \frac{1}{2\eps}$.
\end{corollary}

\begin{proof}
  Since $\delta(\BW_{n})=1$, the upper bounds follow immediately by
  Lemma~\ref{lem:johnson}.  For the equality $\ell(\frac{1}{2}, n) =
  4N$, an easy inductive argument shows that $\abs{L(r, \frac{1}{2})}
  = 4N$ for the received word $r=(\frac{\phi}{2}, \ldots,
  \frac{\phi}{2})\in \C^N$.
\end{proof}

\subsection{Beyond the Johnson Bound}
\label{sec:beyond-johnson-bound}

In this section we prove our main combinatorial bounds on the list
size for Barnes-Wall lattices $\BW_{n} \subseteq \G^{N}$.  Our main
result is that the list size at \rsd $(1-\epsilon)$ is
$(1/\epsilon)^{O(n)} = N^{O(\log (1/\epsilon))}$ for any $\epsilon >
0$.  The proof strategy is inductive, and is based on a careful
partitioning of the lattice vectors in the list according to the
distances of their left and right halves from the respective halves of
the received word.  Intuitively, the larger the distance on one half,
the smaller the distance on the other (Lemma~\ref{lem:delta-relation}
above makes this precise).  The total list size can therefore be
bounded using list bounds for various carefully chosen distances at
lower dimensions.  Our analysis relies on a $\poly(N)$ list-size bound
for \rsd $\frac34$, which in turn relies on a $\poly(N)$ bound for
\rsd $\frac58$.  We first prove these simpler bounds, also using a
partitioning argument.  (Note that the concrete constants appearing
below are chosen to simplify the analysis, and are likely not
optimal.)

\begin{lemma}
  \label{lem:5/8}
  For any integer $n \geq 0$, we have $\ell(\frac{5}{8}, n)\leq 4\cdot
  24^n$.
\end{lemma}

\begin{proof}
  The claim is clearly true for $n=0$, so suppose $n \geq 1$ with
  $N=2^{n}$.  Let $r=[r_0, r_1]\in \C^N$ with $r_0, r_{1} \in
  \C^{N/2}$ be an arbitrary received word, and let $w=[u, u+\phi v]
  \in L(r, \frac58)$ for $u, v \in \BW_{n-1}$.  Let $\eta =
  \delta(r,w) \leq \frac58$, $\eta_0=\delta(r_0, u)$ and
  $\eta_1=\delta(\frac{1}{\phi}(r_1-u), v)$.

  Note that $\eta = \frac 12 (\delta(r_{0}, u) + \delta(r_{1}, u+
  \phi v)) \leq \frac58$.  Without loss of generality, we can assume
  that $\eta_{0} = \delta(r_{0}, u) \leq \frac58$.  For if not, then
  we would have $\delta(r_{1}, u+\phi v) \leq \frac58$, and since $[a,
  b]\in \BW_{n}$ implies $[b, a]\in \BW_{n}$ for $a,b \in \G^{N/2}$,
  we could instead work with the received word $r'=[r_{1}, r_{0}]$ and
  $w' = [u+\phi v, u] \in L(r', \frac58)$.  This incurs a factor of at
  most $2$ in the total list size, which we account for in the
  analysis below.

  Assuming $\eta_{0} \leq \frac58$, we now split the analysis into two
  cases: $\eta_0\in [0, \frac{5}{12})$, and $\eta_0\in [\frac{5}{12},
  \frac5{8}]$.  By Lemma~\ref{lem:delta-relation}, these cases
  correspond to $\eta_1 \leq \frac58$ and $\eta_1 \leq \frac{5}{12}$,
  respectively.  Since $u \in L(r_{0}, \eta_{0})$ and $v \in L(
  \frac{1}{\phi}(r_{1}-u), \eta_{1})$, after incorporating the factor
  of $2$ from the argument above we have (where for conciseness we
  write $\ell(\eta)$ for $\ell(\eta,n-1)$):
  \begin{align*}
    \label{eq:5/12}
    \ell(\tfrac{5}{8}, n) &\leq 2 \cdot \parensfit{
      \ell(\tfrac{5}{12}) \cdot \ell(\tfrac{5}{8}) +
      \ell(\tfrac{5}{8}) \cdot \ell(\tfrac{5}{12}) } \\
    &= 4\cdot \ell(\tfrac{5}{12}) \cdot \ell(\tfrac{5}{8}) \\
    &\leq 4 \cdot 6 \cdot \ell(\tfrac{5}{8}) \\
    &\leq 24^{n} \cdot \ell(\tfrac{5}{8}, 0),
  \end{align*}
  where the penultimate inequality is by Corollary~\ref{cor:johnson},
  and the final one is by unwinding the recurrence.
\end{proof}

\begin{lemma}
  \label{lem:3/4}
  For any integer $n \geq 0$, we have $\ell(\frac{3}{4}, n) \leq 4
  \cdot 24^{2n}$.
\end{lemma}

\begin{proof}
  The claim is clearly true for $n=0$, so suppose $n \geq 1$; we
  proceed by induction on $n$.  Define the same notation as in the
  proof of Lemma~\ref{lem:5/8}, using \rsd bound $\frac34$ instead of
  $\frac58$.

  As before, we assume that $\eta_0\leq \frac34$ and account for the
  accompanying factor of $2$ in the list size.  This time we split the
  analysis into three cases: $\eta_0\in [0, \frac14)$, $\eta_0\in
  [\frac14, \frac58)$, and $\eta_0\in [\frac58, \frac34]$.  By
  Lemma~\ref{lem:delta-relation}, these correspond to $\eta_1\leq
  \frac34$, $\eta_1 \leq \frac58$, and $\eta_1 \leq \frac{7}{16}$,
  respectively.

  For conciseness, in the calculation below we write $\ell(\eta)$ for
  $\ell(\eta,n-1)$.  Using Corollary~\ref{cor:johnson},
  Lemma~\ref{lem:5/8}, and the inductive hypothesis, we have
  \begin{align*}
    \ell(\tfrac34, n) &\leq 2\cdot \parensfit{ \ell(\tfrac14) \cdot
      \ell(\tfrac34) + \ell(\tfrac{5}{8})\cdot \ell(\tfrac{5}{8}) +
      \ell(\tfrac34) \cdot \ell(\tfrac{7}{16}) } \\
    &\leq 2 \cdot (2+8) \cdot \ell(\tfrac34) +
    2\cdot \ell(\tfrac{5}{8})^{2} \\
    &\leq 20 \cdot 4 \cdot 24^{2(n-1)} + 32 \cdot 24^{2(n-1)} \\
    &\leq 4 \cdot 24^{2n}. \qedhere
  \end{align*}
\end{proof}

We are now ready to prove our main combinatorial bound. 


\begin{proof}[Proof of Theorem~\ref{thm:listbound}]
  We need to show that $\ell(1-\epsilon, n) \leq 4\cdot
  (1/\epsilon)^{16n}$ for any $n \geq 0$ and $\epsilon > 0$;
  obviously, we can assume $\epsilon \leq 1$ as well.  The claim is
  clearly true for $n=0$.  We proceed by induction on $n$; namely, we
  assume that for all $\gamma > 0$ it is the case that $\ell(1-\gamma,
  n-1) \leq 4\cdot (1/\gamma)^{16(n-1)}$.  Define the same notation as
  in the proof of Lemma~\ref{lem:5/8}, using \rsd bound $1-\epsilon$
  instead of $\frac58$.

  As in earlier proofs, we assume that $\eta_0\leq 1-\epsilon$ and
  account for the accompanying factor of $2$ in the list size.  We
  split the analysis into 3 cases: $\eta_0 \in [0, \frac12-\eps)$,
  $\eta_0\in [\frac12-\eps, 1-\frac{3\eps}{2})$, and $\eta_0\in
  [1-\frac{3\eps}{2}, 1-\eps]$.  By Lemma~\ref{lem:delta-relation},
  these correspond to $\eta_{1} \leq 1-\epsilon$, $\eta_{1} \leq
  \frac34-\frac{\epsilon}{2} < \frac34$, and $\eta_{1} \leq \frac12
  - \frac{\epsilon}{4}$, respectively.

  For conciseness, in the calculation below we write $\ell(\eta)$ for
  $\ell(\eta,n-1)$.  Using Corollary~\ref{cor:johnson},
  Lemma~\ref{lem:3/4}, and the inductive hypothesis, it follows that
  $\ell(1-\epsilon,n)$ is bounded by
  \begin{align*}
    \MoveEqLeft 2 \left( \ell(1-\eps)
      \ell(\tfrac12-\eps) + \ell(1-\eps)
      \ell(\tfrac12-\tfrac{\epsilon}{4}) +
      \ell(1-\tfrac{3\eps}{2}) \ell(\tfrac34) \right) \\
    & \leq 2 \ell(1-\eps) (\tfrac{1}{2\epsilon} +
    \tfrac{2}{\epsilon}) + 2 \ell(1-\tfrac{3\eps}{2}) \cdot 4
    \cdot 24^{2(n-1)} \\
    &= \tfrac{5}{\eps} \cdot \ell(1-\epsilon)
    + 8 \cdot 24^{2(n-1)} \cdot \ell(1-\tfrac{3\epsilon}{2}) \\
    &\leq \tfrac{20}{\eps} \cdot (\tfrac{1}{\eps})^{16(n-1)}
    +32 \cdot 24^{2(n-1)} \cdot (\tfrac{2}{3\eps})^{16(n-1)}\\
    &= (\tfrac{1}{\eps})^{16(n-1)} \cdot (\tfrac{20}{\eps}+32\cdot
    (24^{2} \cdot (\tfrac{2}{3})^{16})^{(n-1)}) \\
    &\leq (\tfrac{1}{\eps})^{16(n-1)} \cdot (\tfrac{52}{\epsilon}) \\
    &\leq 4\cdot(\tfrac{1}{\eps})^{16n}
  \end{align*}
  when $\eps \leq \frac{4}{5}$. If $\eps \in (\frac{4}{5},1]$ then
  $\ell(1-\eps, n) = 1 \leq 4\cdot(\tfrac{1}{\eps})^{16n}$, and the
  proof is complete.
\end{proof}

Notice that in the above proof, it is important to use an upper bound
like $\eta_{0} \leq 1-\frac{3\epsilon}{2}$ in one of the cases, so
that the factor $(\frac{2}{3})^{16(n-1)}$ from the inductive list
bound can cancel out the corresponding factor of $24^{2(n-1)}$ for the
corresponding \rsd bound $\eta_{1} \leq \frac34$.  This allows the
recurrence to be dominated by the term \[ \ell(1-\epsilon) \cdot
\ell(\tfrac12 - \tfrac{\epsilon}{4}) = O(\tfrac{1}{\epsilon}) \cdot
\ell(1-\epsilon), \] yielding a solution of the form
$(1/\epsilon)^{O(n)}$.

\ifnum\conf=1
Due to space restrictions, we state and prove our lower bounds in
Appendix~\ref{sec:lower-bounds}.

\else

\subsection{Lower Bounds}
\label{sec:lower-bounds}


For our lower bounds we make use of a relationship between
Barnes-Wall lattices and Reed-Muller codes, and then apply known lower
bounds for the latter.

\ifnum\conf=2
\begin{fact}[{\cite[\S IV.B]{Forney88a}}]
  \label{fact:BW-RM}
  $ \BW_{n}=\{ \sum_{d=0}^{n-1} \phi^d  c_d + \phi^n  c_n,~~ 
  c_d\in \RM{d}{n}, 0\leq d\leq n-1;$ $  c_n\in \G^{N} \}, $ where the embedding
  of $\F_2$ into $\C$ is given by $0\mapsto 0$ and $1\mapsto 1$.  In
  particular, any codeword $c_d\in \RM{n}{d}$ gives rise to a lattice
  point $\phi^d \cdot c_d \in \BW_{n}$,
\end{fact}
\else
\begin{fact}[{\cite[\S IV.B]{Forney88a}}]
  \label{fact:BW-RM}
  \[ \BW_{n}=\setfit{ \sum_{d=0}^{n-1} \phi^d \cdot c_d + \phi^n \cdot c_n, \text{ with }
  c_d\in \RM{d}{n}, 0\leq d\leq n-1, \text{ and }  c_n\in \G^{N} }, \] where the embedding
  of $\F_2$ into $\C$ is given by $0\mapsto 0$ and $1\mapsto 1$.  In
  particular, any codeword $c_d\in \RM{n}{d}$ gives rise to a lattice
  point $\phi^d \cdot c_d \in \BW_{n}$,
\end{fact}
\fi

\begin{fact}[{\cite[Chap.~13, \S 4]{MS}}]
  \label{fact:rmdist}
  \ 
  \begin{enumerate}
  \item The minimum distance of $\RM{d}{n}$ is $2^{n-d}$.  In
    particular, the characteristic vector $c_{V} \in \F_{2}^{2^{n}}$
    of any subspace $V \subseteq \F_2^n$ of dimension $k\geq n-d$ is a
    codeword of $\RM{d}{n}$.

    (The characteristic vector $c_{S} \in \F_{2}^{2^{n}}$ of a set $S
    \subseteq \F_{2}^{n}$ is defined by indexing the coordinates of
    $\F_{2}^{2^{n}}$ by elements $\alpha \in \F_{2}^{n}$, and letting
    $(c_{S})_{\alpha} = 1$ if and only if $\alpha \in S$.)
  \item There are $2^d \cdot \prod\limits_{i=0}^{n-d-1}
    \cfrac{2^{n-i}-1}{2^{n-d-i}-1} > 2^{d(n-d)}$ subspaces of
    dimension $n-d$ in $\F_2^n$.
  \end{enumerate}
\end{fact}

\begin{proof}[Proof of Theorem~\ref{thm:lower-bound}]
  Let $k \geq 0$ be an integer such that $2^n \eps \leq 2^k\leq
  2^{n+1} \eps$.  Let the received word be $r =\phi^k \cdot [1, 0,
  \ldots, 0] \in \G^{N}$, where we assume that the first coordinate is
  indexed by $0^n\in \F_2^n$.  By Fact~\ref{fact:rmdist} and
  Fact~\ref{fact:BW-RM}, for any subspace $H\subseteq \F_2^n$ of
  dimension $n-k$, we have $\phi^k \cdot c_H \in \BW_{n}$.  Notice
  that 
 \ifnum\conf=2
 \begin{align*}
  \length{r-\phi^k \cdot c_H }^2 &=\abs{\phi^k}^2 \cdot
  \length{c_H- [1, 0,\ldots, 0]}^{2}\\
  &=2^k \cdot (2^{n-k}-1)\\
  &=2^n-2^k \\
  &\leq 2^n(1-\eps).
  \end{align*}
 \else
  \[ \length{r-\phi^k \cdot c_H }^2=\abs{\phi^k}^2 \cdot
  \length{c_H- [1, 0,\ldots, 0]}^{2}=2^k \cdot (2^{n-k}-1)=2^n-2^k
  \leq 2^n(1-\eps). \]
  \fi
   By Fact~\ref{fact:rmdist}, there are at least
  $2^{k(n-k)}\geq 2^{(n-\log \frac{1}{\eps})\log \frac{1}{2\eps}}$
  subspaces $H\subset \F_2^n$ of dimension $n-k$, which completes the
  proof.
\end{proof}

\fi 


\section{List-Decoding Algorithm}
\label{sec:list-decod-alg}

In this section we prove Theorem~\ref{thm:bwlist-alg} by giving a
list-decoding algorithm that runs in time polynomial in the list size;
in particular, by Theorem~\ref{thm:listbound} it runs in time
$N^{O(\log(1/\epsilon))}$ for \rsd $(1-\epsilon)$ for any fixed
$\epsilon > 0$.  The runtime and error tolerance are optimal (up to
polynomial overhead) in the sense that the list size can be
$N^{\Omega(\log(1/\epsilon))}$ by Theorem~\ref{thm:lower-bound}, and
can be super-polynomial in $N$ for \rsd $1$ or more.

The list-decoding algorithm is closely related to the highly parallel
Bounded Distance Decoding algorithm of Micciancio and
Nicolosi~\cite{MN08}, which 
outputs the unique lattice point within \rsd $\eta<\frac14$ of the
received word (if it exists).
In particular, both algorithms work by recursively (and independently)
decoding four words of dimension $N/2$ that are derived from the
received word, and then combining the results appropriately.  In our
case, the runtime is strongly influenced by the sizes of the lists
returned by the recursive calls, and so the combinatorial bounds from
Section~\ref{sec:combinatorial-bounds} are critical to the runtime
analysis.

We need the following easily-verified fact regarding the symmetries
(automorphisms) of $\BW_{n}$.

\begin{fact}
  \label{fact:automorphisms}
  For $N=2^{n}$, the linear transformation $\calt: \C^N \to \C^N$
  given by $\calt([u,v])= \frac{\phi}{2} \cdot [u+v, u-v]$ is a
  distance-preserving automorphism of $\BW_{n}$, namely
  $\calt(\BW_{n}) = \BW_{n}$ and $\delta(x)=\delta(\calt(x))$ for all
  $x \in \C^{N}$.
\end{fact}

\newcommand{\ldbw}{\textsc{ListDecodeBW}}

\begin{algorithm}[h]
  \caption{\ldbw: List-decoding algorithm for Barnes-Wall lattices.}
  \begin{algorithmic}[1]
    \REQUIRE $r \in \C^N$ (for $N=2^{n}$) and $\eta \geq 0$.

    \ENSURE The list $L(r, \eta) \subset \BW_{n}$.

    \IF{$n=0$} \STATE output $L(r,\eta) \subset \G$ by
    enumeration. \ENDIF

    \STATE parse $r = [r_{0}, r_{1}]$ for $r_{0}, r_{1} \in \C^{N/2}$,
    and let $r_{+} = \frac{\phi}{2}(r_{0}+r_{1})$ and $r_{-} =
    \frac{\phi}{2}(r_{0}-r_{1})$, so $[r_{+},r_{-}] = \calt(r)$.

    \FORALL{$j \in \set{0,1,+,-}$}

    \STATE let $L_{j} = \ldbw(r_{j}, \eta)$.

    \ENDFOR

    \STATE \label{item:combine} for each $(b,s) \in \bit \times
    \set{+,-}$ and each pair $(w_{b}, w_{s}) \in L_{b} \times L_{s}$,
    compute the corresponding candidate vector $w=[w_{0},w_{1}] \in
    \BW_{n}$ as the appropriate one of the following:
    \begin{align*}
      [w_0,\tfrac{2}{\phi} w_{+}-w_0] \quad &, \quad
      [w_0, w_0-\tfrac{2}{\phi}w_{-}], \\
      [\tfrac{2}{\phi}w_{+}-w_1, w_1] \quad &, \quad
      [\tfrac{2}{\phi}w_{-}+w_1, w_1].
    \end{align*}

    \RETURN the set $L$ of all the candidate vectors $w$ such that
    $\delta(r,w) \leq \eta$. \label{item:output}
  \end{algorithmic}
  \label{alg:bwlist}
\end{algorithm}


\begin{proof}[Proof of Theorem \ref{thm:bwlist-alg}]
  We need to show that on input $r\in \C^N$ and $\eta\geq 0$,
  Algorithm~\ref{alg:bwlist} runs in time $O(N^2) \cdot
  \ell(\eta,n)^{2}$ and outputs $L=L(r, \eta)$.

  We first prove correctness, by induction.  The algorithm is clearly
  correct for $n=0$; now suppose that $n \geq 1$ and the algorithm is
  correct for $n-1$.  Adopt the notation from
  Algorithm~\ref{alg:bwlist}, and let $w = [w_{0}, w_{1}] \in
  L(r,\eta)$ for $w_{0}, w_{1} \in \BW_{n-1}$ be arbitrary.  Since
  $\delta(w,r) \leq \eta$, we have $\delta(r_0, w_0)\leq \eta$ or
  $\delta(r_1, w_1)\leq \eta$ or both, so $w_0\in L(r_0, \eta)$ or
  $w_1\in L(r_1, \eta)$ or both.  The same is true about the
  corresponding vectors after applying the automorphism $\calt$.
  Namely, letting $[w_{+}, w_{-}] = \calt(w) \in \BW_{n}$ for $w_{+},
  w_{-} \in \BW_{n-1}$, we have $[w_{+}, w_{-}] \in L([r_{+},r_{-}],
  \eta)$ and so $w_{+}\in L(r_{+}, \eta)$ or $w_{-}\in L(r_{-}, \eta)$
  or both.

  By the inductive hypothesis and the above observations, we will have
  $(w_{b}, w_{s}) \in L_{b} \times L_{s}$ for at least one choice of
  $(b,s) \in \bit \times \set{+,-}$.  The algorithm calculates the
  vector $w=[w_{0},w_{1}]$ as a candidate, simply by solving for
  $w_{0}, w_{1}$ using $w_{b}, w_{s}$ and the definition of $\calt$.
  Therefore, $w$ will appear in the output list~$L$.  And because $L
  \subseteq L(r,\eta)$, the claim follows.

  We now analyze $T(n)$, the number of operations over $\C$ for an
  input of dimension $N=2^{n}$, which is easily seen to satisfy the
  recurrence
  \ifnum\conf=2
  \begin{align*}
    T(n) &\leq 4T(n-1)+ 4\cdot \ell(\eta, n-1)^2 \cdot O(2^{n-1})\\
    &\leq 4^n \cdot T(0) 
    + \sum_{i=1}^{n} 4^i \cdot \ell(\eta,
    n-i)^{2} \cdot O(2^{n-i}) \\
    &\leq O(N^{2}) + O(2^{n}) \cdot \ell(\eta,n-1)^{2} \cdot
    \sum_{i=1}^{n} 2^{i} \\
    &=O(N^2) \cdot \ell(\eta, n-1)^{2}. \qedhere
  \end{align*}
  \else
   \begin{align*}
    T(n) &\leq 4T(n-1)+ 4\cdot \ell(\eta, n-1)^2 \cdot O(2^{n-1})
    \leq 4^n \cdot T(0) + \sum_{i=1}^{n} 4^i \cdot \ell(\eta,
    n-i)^{2} \cdot O(2^{n-i}) \\
    &\leq O(N^{2}) + O(2^{n}) \cdot \ell(\eta,n-1)^{2} \cdot
    \sum_{i=1}^{n} 2^{i} =O(N^2) \cdot \ell(\eta, n-1)^{2}. \qedhere
  \end{align*}
  \fi
\end{proof}


\begin{remark} 
  We note that the above algorithm, like the unique decoder
  of~\cite{MN08}, can be easily parallelized.  The parallel time on
  $p$ processors (counting the number of operations in $\C$) satisfies
  the recurrence \ifnum\conf=2
  \begin{eqnarray*}
    T(n, p)=\begin{cases}
      T(n), \mbox{ if } n=0 \mbox{ or } p<4\\
      T(n-1, p/4)\\
      \quad {} +O(N \cdot \ell(\eta,n-1)^2/p + \log N) \text{ o.w.,}
    \end{cases}
  \end{eqnarray*}
  \else
  \begin{eqnarray*}
    T(n, p)=\begin{cases}
      T(n)& \mbox{ if } n=0 \mbox{ or } p<4\\
      T(n-1, p/4)+O(N \cdot \ell(\eta,n-1)^2/p + \log N) & \text{otherwise,}
    \end{cases}
  \end{eqnarray*}  
  \fi where $T(n)$ is the sequential time computed in
  Theorem~\ref{thm:bwlist-alg}.  This is because it takes $O(N \cdot
  \ell(\eta,n-1)^2/p)$ time per processor to combine the lists in
  Step~\ref{item:combine} of the algorithm, and computing the
  $\ell(\eta,n-1)^2$ distances in Step~\ref{item:output} requires
  computing sums of $N$ terms in $\C$, and takes a total of $O(N \cdot
  \ell(\eta, n-1)^2/p + \log N)$ parallel time.  Notice that when $p
  \geq N^{2} \cdot \ell(\eta, n-1)^2$, the algorithm runs in only
  polylogarithmic $O(\log^2 N)$ parallel time.  Note also that when
  the list size $\ell(\eta,n-1)=1$, our analysis specializes exactly
  to that of~\cite{MN08}.
\end{remark}


\section{Discussion and Open Problems}
\label{sec:disc-open-probl}

Some immediate open questions arise from comparison to the results
in~\cite{MN08}.  Motivated by the sequential unique decoder proposed
in~\cite{MN08}, is there a (possibly sequential) list decoder that
runs in time quasilinear in $N$ and the list size, rather than
quadratic?  Also, as asked in~\cite{MN08}, is there an efficient
algorithm for solving the Closest Vector Problem (i.e.,
minimum-distance decoding) on Barnes-Wall lattices?  Note that our
lower bounds do not rule out the existence of such an algorithm.

An important variant of the list-decoding problem for codes is {\em
  local} list decoding.  In this model, the algorithm is required to
run in time polylogarithmic in the block length, and output succinct
representations of all the codewords within a given radius.  Defining
a meaningful notion of local decoding for lattices (and BW lattices in
particular) would require additional constraints, since lattice points
do not in general admit succinct representations (since one needs to
specify an integer coefficient for each basis vector).  While by the
Johnson bound we have a $\poly(n)$ list size for \rsd up to
$1/2-\poly(1/n)$, achieving a meaningful notion of local decoding in
this context would be interesting.

Another interesting direction is to find (or construct) more
asymptotic families of lattices with nice list-decoding properties.
In particular, are there generic operations that when applied to
lattices guarantee good list-decoding properties?  For codes, list
decodability has been shown to behave well under the tensoring and
interleaving operations, as demonstrated in~\cite{GopalanGR09}.  Since
at least tensoring is also well-defined for lattices, understanding
its effect in the context of list decoding is a natural further
direction.

Finally, it would be also interesting and potentially useful to
consider list decoding for norms other than the Euclidean norm, such
as the $\ell_{\infty}$ or $\ell_{0}$ norms.


\section*{Acknowledgments}

We thank Eli Ben-Sasson, Daniele Micciancio, Madhu Sudan, and Santosh
Vempala for helpful discussions and comments.  We thank the anonymous
referees for their helpful suggestions, and for pointing out a bug in
our previous statement of Lemma~\ref{lem:johnson}.


\bibliographystyle{plain}
\bibliography{coding}

\ifnum\conf=1
\appendix

\section{Appendix}
\label{sec:appendix}

\subsection{Missing Proofs}
\label{sec:missing-proofs}
%
%

\begin{proof}[Proof of Lemma~\ref{lem:3/4}]
  The claim is clearly true for $n=0$, so suppose $n \geq 1$; we
  proceed by induction on $n$.  Define the same notation as in the
  proof of Lemma~\ref{lem:5/8}, using \rsd bound $\frac34$ instead of
  $\frac58$.

  As before, we assume that $\delta_0\leq \frac34$ and account for the
  accompanying factor of $2$ in the list size.  This time we split the
  analysis into three cases: $\delta_0\in [0, \frac14)$, $\delta_0\in
  [\frac14, \frac58)$, and $\delta_0\in [\frac58, \frac34]$.  By
  Lemma~\ref{lem:delta-relation}, these correspond to $\delta_1\leq
  \frac34$, $\delta_1 \leq \frac58$, and $\delta_1 \leq \frac{7}{16}$,
  respectively.

  Using Corollary~\ref{cor:johnson} and Lemma~\ref{lem:5/8}, we therefore
  have
  \begin{align*}
    \ell(\tfrac34, n) &\leq 2\cdot \parensfit{ \ell(\tfrac14,
      n-1)\cdot\ell(\tfrac34, n-1)+\ell(\tfrac{5}{8}, n-1)\cdot
      \ell(\tfrac{5}{8}, n-1)+\ell(\tfrac34, n-1)\cdot
      \ell(\tfrac{7}{16}, n-1) } \\
    &\leq 2 \cdot (2 +8) \cdot \ell(\tfrac34, n-1) +
    2\cdot \ell(\tfrac{5}{8}, n-1)^{2} \\
    &\leq 20 \cdot 4 \cdot 24^{2(n-1)} + 32 \cdot 24^{2(n-1)} \\
    &\leq 4 \cdot 24^{2n}. \qedhere
  \end{align*}
\end{proof}

\subsection{Lower Bounds}
\label{sec:lower-bounds}


For our lower bounds we make use of a relationship between
Barnes-Wall lattices and Reed-Muller codes, and then apply known lower
bounds for the latter.

\begin{fact}[{\cite[\S IV.B]{Forney88a}}]
  \label{fact:BW-RM}
  \[ \BW_{n}=\{ \sum_{d=0}^{n-1} \phi^d c_d+\phi^n c_n : c_d\in
  \RM{d}{n}, 0\leq d\leq n-1; c_n\in \G^{N}\}, \] where the embedding
  of $\F_2$ into $\C$ is given by $0\mapsto 0$ and $1\mapsto 1$.  In
  particular, any codeword $c_d\in \RM{n}{d}$ gives rise to a lattice
  point $\phi^d \cdot c_d \in \BW_{n}$,
\end{fact}

\begin{fact}[{\cite[Chap.~13, \S 4]{MS}}]
  \label{fact:rmdist}
  \ 
  \begin{enumerate}
  \item The minimum distance of $\RM{d}{n}$ is $2^{n-d}$.  In
    particular, the characteristic vector $c_{V} \in \F_{2}^{2^{n}}$
    of any subspace $V \subseteq \F_2^n$ of dimension $k\geq n-d$ is a
    codeword of $\RM{d}{n}$.

    (The characteristic vector $c_{S} \in \F_{2}^{2^{n}}$ of a set $S
    \subseteq \F_{2}^{n}$ is defined by indexing the coordinates of
    $\F_{2}^{2^{n}}$ by elements $\alpha \in \F_{2}^{n}$, and letting
    $(c_{S})_{\alpha} = 1$ if and only if $\alpha \in S$.)
  \item There are $2^d \cdot \prod\limits_{i=0}^{n-d-1}
    \cfrac{2^{n-i}-1}{2^{n-d-i}-1} > 2^{d(n-d)}$ subspaces of
    dimension $n-d$ in $\F_2^n$.
  \end{enumerate}
\end{fact}

\begin{proof}[Proof of Theorem~\ref{thm:lower-bound}]
  Let $k \geq 0$ be an integer such that $2^n \eps \leq 2^k\leq
  2^{n+1} \eps$.  Let the received word be $r =\phi^k \cdot [1, 0,
  \ldots, 0] \in \G^{N}$, where we assume that the first coordinate is
  indexed by $0^n\in \F_2^n$.  By Fact \ref{fact:rmdist} and
  Fact~\ref{fact:BW-RM}, for any subspace $H\subseteq \F_2^n$ of
  dimension $n-k$, we have $\phi^k \cdot c_H \in \BW_{n}$.  Notice
  that $\length{r-\phi^k \cdot c_H }^2=\abs{\phi^k}^2 \cdot
  \length{c_H- [1, 0,\ldots, 0]}^{2}=2^k \cdot (2^{n-k}-1)=2^n-2^k
  \leq 2^n(1-\eps)$.  By Fact~\ref{fact:rmdist}, there are at least
  $2^{k(n-k)}\geq 2^{(n-\log \frac{1}{\eps})\log \frac{1}{2\eps}}$
  subspaces $H\subset \F_2^n$ of dimension $n-k$, which completes the
  proof.
  %
\end{proof}


\fi

\end{document}